\documentclass{amsart}

\numberwithin{equation}{section}

\usepackage{amssymb}
\usepackage{enumerate, xspace}
\usepackage{dsfont}
\usepackage[bottom,first]{draftcopy}
\usepackage{changebar}
\usepackage[colorlinks]{hyperref}

\newtheorem{thmm}{Theorem}
\newtheorem{thm}{Theorem}[section]
\newtheorem{lem}[thm]{Lemma}
\newtheorem{cor}[thm]{Corollary}

\newtheorem{prop}[thm]{Proposition}

\newtheorem{rem}[thm]{Remark}

\newcommand\cC{{\mathcal C}}

\newcommand\cE{{\mathcal E}}

\newcommand\cH{{\mathcal H}}

\newcommand\cL{{\mathcal L}}
\newcommand\cO{{\mathcal O}}
\newcommand\cM{{\mathcal M}}

\newcommand\bE{{\mathbb E}}
\newcommand\bG{{\mathbb G}}
\newcommand\bM{{\mathbb M}}
\newcommand\bN{{\mathbb N}}
\newcommand\bP{{\mathbb P}}

\newcommand\bR{{\mathbb R}}

\newcommand\bZ{{\mathbb Z}}

\newcommand\ba{{\boldsymbol{a}}}
\newcommand\bbeta{{\boldsymbol{\beta}}}

\newcommand\ve{\varepsilon}

\newcommand\vf{\varphi}

\newcommand\Id{{\mathds{1}}}
\newcommand\id{id}

\newcommand{\bi}{{\boldsymbol i}}

\begin{document}

\title{Energy transfer in a fast-slow Hamiltonian system}
\author{Dmitry Dolgopyat}
\address{Dmitry Dolgopyat\\
Department of Mathematics\\
University of Maryland\\
4417 Mathematics Bldg,  College Park,  MD 20742, USA}
\email{{\tt dmitry@math.umd.edu}}
\urladdr{http://www.math.umd.edu/\~\ \hskip-4pt dmitry}
\author{Carlangelo Liverani}
\address{Carlangelo Liverani\\
Dipartimento di Matematica\\
II Universit\`{a} di Roma (Tor Vergata)\\
Via della Ricerca Scientifica, 00133 Roma, Italy.}
\email{{\tt liverani@mat.uniroma2.it}}
\urladdr{http://www.mat.uniroma2.it/\~\ \hskip-4pt liverani}
\date{\today.}% {\bf File: {\jobname}.tex.}}
\begin{abstract}
We consider a finite region of a lattice of weakly interacting geodesic flows on manifolds of negative curvature and we show that, when rescaling the interactions and the time appropriately, the energies of the flows evolve according to a non linear diffusion equation. This is a first step toward the derivation of macroscopic equations from a Hamiltonian microscopic dynamics in the case of weakly coupled systems.
\end{abstract}
\thanks{We thank Gabriel Paternain for suggesting to us reference \cite{CS}. C.Liverani acknowledges the partial support of the European Advanced Grant  Macroscopic Laws and Dynamical Systems (MALADY) (ERC AdG 246953). Both authors warmly thank CIRM, Marseille, that fostered the beginning of this work by financing our {\em research in pairs}. Finally, it is a pleasure to thank the  Fields Institute, Toronto, where the paper was finished.}
\maketitle

%%%PAPER
\section{Introduction}
One of the central problems in the study of non-equilibrium statistical physics
is the derivation of transport equations for conserved quantities, in particular energy transport, 
from first principles, (see \cite{BLR}, and references therein, or \cite{Sp}, for a more general discussion on the derivation of macroscopic equations from microscopic dynamics). 

Lately several results have appeared trying to bring new perspective to the above problem in a collective effort to attack the problem from different points of views. Let us just mention, as examples, papers considering stochastic models \cite{BBO, BOS, B}, approaches starting from kinetic equations or assuming extra hypotheses \cite{ALS, LS, BK1} 
or papers trying to take advantage of the point of view and 
results developed in the field of Dynamical Systems 
\cite{EY, DKL,DL1,DL2, BK2, CE, Ru}. 
This paper belongs to the latter category but it is closely related to results obtained for stochastic models (e.g., \cite{LO}).

We consider a {\em microscopic} dynamics determined by a (classical) Hamiltonian describing a finite number of weakly interacting strongly chaotic systems and we explore the following strategy to derive a macroscopic evolution: first one looks at times for which we have an effective energy exchange between interacting systems, then takes the limit for the strength of the interaction going to zero and hopes to obtain a self-contained equation describing the evolution of the energies only. We call such an equation {\em mesoscopic} since most of the degrees of freedom have been averaged out. 
Second, one performs on such a mesoscopic equation a thermodynamic limit to obtain a {\em macroscopic} evolution. In particular, one can consider a scaling limit of the diffusive type in order to obtain a non linear heat equation as in the case of the so called {\em hydrodynamics limit} for particle systems, see \cite{GPV, Va} for more details. A similar strategy has been carried out, at a heuristic level, in \cite{GG1, GG2}.

The {\em first step} of such a program is accomplished in this paper. 
It is interesting to note that the mesoscopic equation that we obtain seems to have some very natural and universal structure since it holds also when starting from different models. Indeed, essentially the same equation is obtained in \cite{LO} 
for a system of coupled nonlinear oscillators
in the presence of an energy preserving randomness. In addition, such an equation is almost identical to the one studied in \cite{Va} apart from the necessary difference that the diffusion is a degenerate one. Indeed, since it describes the evolution of energies, and energies are positive, the diffusion coefficients must necessarily be zero when one energy is zero.

Since, due to the weak interaction, the energies vary very slowly, once the time is rescaled so that the energies evolve on times of order one all the other variables will evolve extremely fast. Thus our result is an example of {\em averaging theory} for slow-fast systems. Yet, in our case the currents have zero average which means that standard averaging theory (such as, e.g. \cite{FW}) cannot suffice. It is necessary to look at longer times when the fluctuations play a fundamental role. The study of such longer times can in principle be accomplished thanks to the theory developed in \cite{Do2}. 

Unfortunately, the results in \cite{Do2} do not apply directly and we are forced to a roundabout in order to obtain the wanted result. Not surprisingly, the trouble takes place at low energies. We have thus to investigate with particular care the behavior of the system at low energies. In particular, we prove that the probability for any particle to reach zero energy, in the relevant time scale, tends to zero.

The structure of the paper is as follows: section \ref{sec:model} contains the precise description of the microscopic model and the statement of the results. Section \ref{sec:heu} describes the logic of the proof at a non technical level and points out the technical difficulties that must be overcame to make the argument rigorous. In the following section we show how to modify the dynamics at low energies in such a way that existing results can be applied. Then, in section \ref{sec:existence}, we investigate the modified dynamics and show that its accumulation points satisfy a mesoscopic equation of the wanted type. 
In section \ref{sec:compute} we compute explicitly the properties of the coefficients of the limit equation for the modified dynamics and in section \ref{sec:structure} we use this knowledge to show that the equation has a unique solution, hence the modified process converges to this solution. In section \ref{sec:reach} we discuss the limit equation for the original dynamics under the condition that no particle reaches zero in finite time. The fact that this condition 
holds in our model is proven in section \ref{sec:zero}. 
The paper ends with two appendices. 
In the first, for the reader convenience, some known results 
from the averaging theory for systems with
hyperbolic fast motion are restated in a way suitable for our needs. 
The second appendix
contains some boring, but essential, computations.

\section{The model and the result}\label{sec:model}
For $d\in\bN$, we consider a lattice $\bZ^d$ and a finite connected 
region $\Lambda\subset \bZ^d$. 
Associated to each site in $\Lambda$ we have the cotangent bundle $T^*M$ 
of a $\cC^\infty$ compact Riemannian $d$-dimensional manifold $M$ of 
strictly negative curvature and the associated geodesic flow $g^t$.  
We have then the phase space $\cM=(T^*M)^\Lambda$ and we designate a point 
as $(q_x,p_x)$, $x\in\Lambda$. 
It is well known that the geodesic flows is a Hamiltonian flow. 
If we define $\bi: T^*M\to TM$ to be the natural isomorphism defined by $w(v)=\langle \bi(w), v\rangle_{G}$, $G$ being the Riemannian metric, then the Hamiltonian reads\footnote{By $p_x^2$ we mean $\langle \bi (p_x),\bi(p_x)\rangle_{G(q_x)}=\langle p_x,p_x\rangle_{\tilde G}$ where $\tilde G={\bi_*(G)}$.} $H_0=\sum_{x\in\Lambda} \frac 12 p_x^2$ and the symplectic form is given by $\omega=d\boldsymbol{q}\wedge d\boldsymbol{p}$.\footnote{To be more precise, given the canonical projection $\pi(q,p)=q$, first define the one form, on $T(T^* M)$, $\omega^1_{(q,p)}(\xi)=p(d\pi(\xi))$. Then $\omega:=-d\omega^1$. Given coordinates $\boldsymbol{q}$ on $U\subset M$ and using the coordinates $\boldsymbol{p}$ for the one form $p=\sum_i \boldsymbol{p}_i\,d\boldsymbol{q}_i\in T^*M$, we have $\omega^1=\sum_i \boldsymbol{p}_i\, d\boldsymbol{q}_i$ and $\omega=\sum_i d\boldsymbol{q}_i\wedge d\boldsymbol{p}_i$, as stated.}
Thus, given $x\in\Lambda$, the equations of motion take the form (see \cite[Section 1]{Pa} for more details)
\begin{equation}\label{eq:hami-geo}
\begin{split}
&\dot q_x=\bi(p_x)\,,\\
&\dot p_x=\tilde F(q_x,p_x)\,,
\end{split}
\end{equation}
where the $\tilde F$ is homogeneous in the $p_x$ of degree two.
Note that, by the Hamiltonian structure, $e_x:=\frac 12 p_x^2$ is constant in time for each $x\in \Lambda$. It is then natural to use the variables $(q_x,  v_x, e_x)$, where $ v_x:=(p_x^2)^{-\frac 12}\bi(p_x)$ belongs to the unit tangent bundle $T^1M$ of $M$.\footnote{Clearly $e_x$ is the (kinetic) energy of the geodesic flow at $x$.} We have then the equations
\begin{equation}\label{eq:indip-geo}
\begin{split}
&\dot q_x=\sqrt {2e_x}v_x\,,\\
&\dot v_x=\sqrt {2e_x}F(q_x,v_x)\,,\\
&\dot e_x=0,
\end{split}
\end{equation}
where $F$ is homogeneous of second degree in  $v_x$.

Next we want to introduce a small energy exchange between particles. To describe such an exchange we introduce a symmetric, non constant, function (potential) $V\in\cC^\infty(M^2, \bR)$ and, for each $\ve>0$, consider the flow $g_\ve^t$ determined by the Hamiltonian $H_\ve=\sum_{x\in\Lambda} \frac 12 p_x^2+\frac\ve 2\sum_{|x-y|=1}V(q_x,q_y)$, that is by the equations
\[
\begin{split}
&\dot q_x=\bi(p_x)\,,\\
&\dot p_x=\tilde F(q_x,p_x)-\ve\sum_{|y-x|=1}d_{q_x}V(q_x,q_y).
\end{split}
\]
Or, alternatively,\footnote{In the interacting case one could chose to include the interaction in the energy and define $e_x^\ve:=\frac 12 p_x^2+\frac\ve 4\sum_{|x-y|=1}V(q_x,q_y)$. This is the choice made in \cite{LO}. Yet, in the present context $|e_x-e_x^\ve|\leq |V|_\infty\ve$, hence the actual choice is irrelevant in the limit $\ve\to 0$ and $e_x$ turns out to be computationally simpler.}
\begin{equation}\label{eq:full}
\begin{split}
&\dot q_x=\sqrt{2e_x} v_x\\
&\dot v_x=\sqrt{2e_x} F(q_x,v_x)+\frac{\ve}{\sqrt{2e_x}}\sum_{|y-x|=1}\{  v_x L_x V-\nabla_{q_x}V(q_x,q_y)\}\\
&\dot e_x=-\ve\sqrt{2e_x}\sum_{|x-y|=1}L_xV,
\end{split}
\end{equation}
where $\langle\nabla V, w\rangle_G= d V(w)$ and 
\begin{equation}
\label{GenGF}
L_x=v_x\partial_{q_x}+F(q_x,v_x)\partial_{v_x}
\end{equation}
denotes the generator associated to the geodesic flow of the $x$ particle on $T_1M$.

We will consider {\em random} initial conditions of the following type
\begin{equation}\label{eq:init-cond}
\begin{split}
&\bE(f(q(0),v(0))=\int_{(T_1M)^\Lambda} f(q,v)\rho(q,v) dm, \quad\forall f\in\cC^0((T_1M)^\Lambda,\bR)\\
&e_x(0)=E_x>0,
\end{split}
\end{equation}
where $m$ is the Riemannian measure on $(T_1M)^\Lambda$ and $\rho\in\cC^1$.

Since the currents $L_xV$ have zero average with respect to the microcanonical measure, one expects that it will take a time of order $\ve^{-2}$ in order to see a change of energy of order one.
It is then natural to introduce the process $e_x(\ve^{-2}t)$ and to study the convergence of such a process in the limit $\ve\to 0$. 

Our main result is the following.

\begin{thmm}\label{thm:main} Provided $d\geq 3$, the process $\{e_x(\ve^{-2}t)\}$ defined by \eqref{eq:full} with initial conditions \eqref{eq:init-cond}  converges to a random process $\{\cE_x(t)\}$ with values in $\bR_+^{\Lambda}$ which satisfies the stochastic differential equation
\begin{equation}\label{eq:final-eq}
\begin{split}
&d\cE_x=\sum_{|x-y|=1}\ba(\cE_x,\cE_y)dt+\sum_{|x-y|=1}\sqrt 2\bbeta(\cE_x,\cE_y) dB_{xy}\\
&\cE_x(0)=E_x>0,
\end{split}
\end{equation}
where $B_{xy}$ are standard Brownian motions which are independent except that $$B_{xy}=-B_{yx}.$$ 

The coefficients have the following properties: $\bbeta$ is symmetric and $\ba$ is antisymmetric; $\bbeta\in\cC^0([0,\infty)^2,\bR_+)$ and $\bbeta(a,b)^2= abG(a,b)$ where $G\in\cC^\infty((0,\infty)^2,\bR_+)\cap\cC^1((0,\infty)\times [0,\infty),\bR_+)$ and $G(a,0)=A(2a)^{-\frac 32}$ for some $A>0$. Moreover,\begin{equation}\label{eq:drift-formula}
\ba=(\partial_{\cE_x}-\partial_{\cE_y})\bbeta^2+\frac{d -2}2(\cE_x^{-1}-\cE_y^{-1})\bbeta^2.
\end{equation}
In addition, \eqref{eq:final-eq} has a unique solution and the probability for one energy to reach zero in finite time is zero.
\end{thmm}
\begin{rem} A direct computation shows that the measures with density $h_\beta=\prod_{x\in\Lambda} \cE_x^{\frac{d}2-1} e^{-\beta\cE_x}$ are invariant for the above process for each $\beta\in\bR_+$.
Indeed, using \eqref{eq:drift-formula}, we can write the generator of the process \eqref{eq:final-eq}  in the simple form
\[
\cL=\frac 1{2h_0}\sum_{|x-y|=1}(\partial_{\cE_x}-\partial_{\cE_y})h_0\bbeta^2(\partial_{\cE_x}-\partial_{\cE_y})
\]
from which the reversibility of the generator is evident.
\end{rem}
\begin{rem} The case $d=2$ is harder because the second term in \eqref{eq:drift-formula} 
(which otherwise would give the main contribution at small energies) is zero. 
We believe the result to be still true,\footnote{That is the fact that zero is unreachale.} but a much more detailed (and messy) analysis of \eqref{eq:final-eq} is needed to establish it. As this would considerably increase the length of section \ref{sec:zero} without adding anything really substantial to the paper, we do not pursue such matter. 

\end{rem}
\begin{rem} Note that if we could apply \cite{Va} to perform the hydrodynamics limit, then we would obtain the heat equation. Unfortunately, \eqref{eq:final-eq} does not satisfies the hypotheses of Varadhan's Theorem on several accounts; the most relevant being that the domain where the diffusion takes place is not the all space and $\ba, \bbeta$ vanish one the boundary of the domain. This is unavoidable as the energy is naturally bounded below. Nevertheless, the results of this paper can be considered as a first step along the bumpy road to obtaining the heat equation from a purely mechanical deterministic model.\footnote{One could object that geodesic motion in negative curvature is not really mechanical. Yet, it is possible to construct a bona fide mechanical system which motion is equivalent to a geodesic flow in negative curvature \cite{HM}. In any case, by Maupertuis' principle, any Hamiltonian system can be viewed as a geodesic flow, possibly on a non compact manifold.} 
\end{rem}
\begin{rem}
As a last remark, let us comment on the choice of $\bZ^d$. This is done just to simplify notations: our arguments are of a local nature, hence the structure of $\bZ^d$ does not play any role in the proof.
In particular, one can prove, with exactly the same arguments, the following extension of our result.

Consider a loopless symmetric directed graph $\bG$ determined by the collection of its vertexes $V(\bG)$ and the collection of its directed edges $E(\bG)$.\footnote{Directed means that the edges $e\in E(\bG)$ are ordered pairs $(e_1,e_2)$, $e_i\in V(\bG)$, which is interpreted as an edge going from $e_1$ to $e_2$. Symmetric means that if $(e_1,e_2)\in  E(\bG)$, then $(e_2,e_1)\in E(\bG)$. Loopless that, for each $a\in V(\bG)$, $(a,a)\not\in E(\bG)$. This abstract setting reduces to the previous one if we choose $V(\bG)=\bZ^d$ and $E(\bG)=\{(x,y)\in\bZ^d\times \bZ^d\;:\; |x-y|=1\}$.} 
At each vertex $v\in V(\bG)$ we associate a mixing geodesic flow as before, consider then the Hamiltonian
\[
H_\ve=\sum_{v\in V(\bG)}\frac 12 p_v^2+\frac \ve 2\sum_{(e_1, e_2)\in E(\bG)}V(q_{e_1},q_{e_2}).
\] 
We then have the exact analogous\footnote{In particular the condition $d\geq 3$ refers to the manifolds $M$ not to the lattice or graph.} of Theorem \ref{thm:main} for the variables $\{\cE_v\}_{v\in V(\bG)}$ with the only difference that the limiting equation now reads
\begin{equation}\label{eq:general}
\begin{split}
&d\cE_v=\sum_{(v,w)\in E(\bG)}\ba(\cE_v,\cE_w)dt+\sum_{(v,w)\in E(\bG)}\sqrt 2\bbeta(\cE_v,\cE_w) dB_{(v,w)}\\
&\cE_v(0)=E_v>0,
\end{split}
\end{equation}
where again for each $e\in E(\bG)$, the $B_e$ are independent standard Brownian motions apart form the fact that $B_{(v,w)}=-B_{(w,v)}$.
\end{rem}
An interesting application of the above Remark is the case where $\bG$ 
is a complete graph (i.e. $E(\bG)=\{(v_1,v_2)\;:\:v_1,v_2\in V(\bG)\}$) 
in which case all particles interact with each other.

The rest of the paper is devoted to proving Theorem \ref{thm:main}. Before going in details we explain exactly how the various results we are going to derive are collected together to prove the Theorem. 

\begin{proof}[{\bf Proof of Theorem \ref{thm:main}}] Fix $T>0$ and let $\bP_\ve$ be the probability measure, on the space $\cC^0([0,T], \bR_+^{\Lambda})$, associated to the process $\{e_x(\ve^{-2}t)\}_{t\in [0,T]}$ defined by \eqref{eq:full}, 
$\bP_{\ve,\delta}$ to the one defined by \eqref{eq:fullm}, $\tilde \bP_\delta$ the one associated to the process $\{e^{z(t)}\}$ with $z(t)$ defined by \eqref{eq:limit-e} and $\bP$ the one defined by \eqref{eq:final-eq}. Also, let $\Omega_\delta=\{\tau_\delta\geq T\}$ where $\tau_\delta=\inf\{t\in\bR_+\;:\;\min_{x\in\Lambda}\cE_x(t)\leq \delta\}$.  By construction, for each $F\in\cC^0$, $\bE_{\bP_\ve}(F \Id_{\Omega_\delta})=\bE_{\bP_{\ve,\delta}}(F\Id_{\Omega_\delta})$, $\bE_{\tilde\bP_\delta}(F\, \Id_{\Omega_\delta})=\bE_{\bP}(F\, \Id_{\Omega_\delta})$. 

Proposition \ref{prop:uniqueness} implies that $\bP_{\ve,\delta}\Rightarrow\tilde\bP_\delta$ and, since $\Omega_\delta$ is a continuity set for $\tilde\bP_\delta$, $\lim_{\ve\to 0}\bP_{\ve,\delta}(\Omega_\delta)=\tilde\bP_\delta(\Omega_\delta)=\bP(\Omega_\delta)$. 

Next, Lemma \ref{lem:no-zero},  based on estimate \eqref{DriftLog}, tells us that $\lim_{\delta\to 0}\tilde\bP_\delta(\Omega_\delta^c)=0.$ 
Thus
\[
\lim_{\delta\to 0}\lim_{\ve\to 0}\bP_\ve(\Omega_\delta^c)=0.
\]
Hence $\bP_\ve\Rightarrow \bP.$ 
The informations on the coefficients follow by collecting \eqref{eq:beta-final}, \eqref{eq:drift-formula} (proven in Lemma \ref{CrDriftUp}), Lemmata \ref{lem:rho-tilde} and \ref{lem:rho-asymp}. Finally, the uniqueness follows from standard results on SDE and the unreachability of zero (Lemma \ref{lem:no-zero}).
\end{proof}

\section{Heuristic}\label{sec:heu}
Let us give a sketch of the argument where we ignore all the technical difficulties and perform some daring formal computations.

If we could apply \cite[Theorem 7]{Do2} to equation \eqref{eq:full} we would obtain a limiting process characterized by an equation that, after some algebraic manipulations detailed in section \ref{sec:structure}, reads\footnote{See Appendix \ref{sec:dolgopyat} for a precise statement of the results in \cite{Do2} relevant to our purposes.}
\begin{equation}\label{eq:belief}
d\cE_x=\sum_{|x-y|=1}\ba(\cE_x,\cE_y)dt+\sum_{|x-y|=1}\sqrt 2\bbeta(\cE_x,\cE_y) dB_{xy}
\end{equation}
where $\beta(\cE_x,\cE_y)=\beta(\cE_y,\cE_x)$ is symmetric and $B_{xy}=-B_{yx}$ are independent standard Brownian motions. The marginal of the Gibbs measure on the energy variables reads
\[
d\mu_\beta=\prod_x\cE_x^{\frac{d}2-1}e^{-\beta\cE_x} d\cE_x=:h_\beta \wedge_x d \cE_x,
\]
for each $\beta\in[0,\infty)$. Hence we expect such a measure to be invariant for  \eqref{eq:belief}. Even more, on physical grounds (see Lemma \ref{lem:reversi}) 
one expects the process \eqref{eq:belief} to be reversible with respect 
to this measures. A straightforward computation shows that 
the generator associated to the above SDE reads
\[
\cL=\sum_{|x-y|=1}\ba_{xy}\partial_{\cE_x}+\frac 12\sum_{|x-y|=1}\bbeta_{xy}^2(\partial_{\cE_x}-\partial_{\cE_y})^2,
\]
where $\ba_{xy}=\ba(\cE_x,\cE_y)$, $\bbeta_{xy}=\bbeta(\cE_x,\cE_y)$. The adjoint with respect to $\mu_0$ reads
\[
\begin{split}
\cL^*=&\sum_{|x-y|=1}\left\{-\ba_{xy}+\frac{d+1}2(\cE_x^{-1}-\cE_y^{-1})\bbeta_{xy}^2+(\partial_{\cE_x}-\partial_{\cE_y})\bbeta_{xy}^2\right\}\partial_{\cE_x}\\
&+\frac 12\sum_{|x-y|=1}\bbeta_{xy}^2(\partial_{\cE_x}-\partial_{\cE_y})^2-\frac 1{h_0}\sum_{|x-y|=1}\partial_{\cE_x}(h_0\ba_{xy})\\
&+\frac 1{2h_0}\sum_{|x-y|=1}(\partial_{\cE_x}-\partial_{\cE_y})^2(h_0\bbeta_{xy}).
\end{split}
\]
Computing what it means $\cL=\cL^*$ implies \eqref{eq:drift-formula}.

\begin{rem}Note that, as expected, $\ba_{xy}=-\ba_{yx}$. Thus $d\sum_x\cE_x=0$.
\end{rem}

Going to a bit less vague level of analysis, one must notice that since $\cE_x\geq 0$, the diffusion equation \eqref{eq:belief} must be degenerate at zero, also it is not clear how regular the coefficients $\ba,\bbeta$ are. Hence, a priori, it is not even obvious that such an equation has a solution and, if so, if such a solution is unique. To investigate such an issue it is necessary to obtain some information on the behavior of the coefficients at low energies.

To this end one can use the explicit formula given in \cite[Theorem 7]{Do2} for the diffusion coefficient. This allows to verify that the coefficients are smooth away from zero. An explicit,  but lengthy,  computation yields, for $\cE_x\leq \cE_y$,
\begin{equation}\label{eq:ab-asy}
\begin{split}
&\bbeta_{xy}^2=\frac{ A\cE_x}{\sqrt{2\cE_y}}+\cO\left(\cE_x^{\frac 32}\cE_y^{-1}\right)\\
&\ba_{xy}=\frac{Ad}{2\sqrt{2\cE_y}}+\cO\left(\frac{ \sqrt{\cE_x}}{\cE_y}\right),
\end{split}
\end{equation}
see Lemma \ref{CrDriftUp} for details.
Thus, in particular, $\ba_{xy}\cE_x=\frac{d}2\bbeta_{xy}^2+o(\bbeta_{xy}^2)$.

We will see in section \ref{sec:zero} that such a relation, 
provided $d>2$, suffices to prove that the 
set $\{(\cE_x)\;:\; \prod_x\cE_x=0\}$ is unreachable 
and hence to insure that equation \eqref{eq:belief} has a unique solution.

In the rest of the paper we show how to make rigorous the above line of reasoning. 

\section{A modified dynamics}

Since the geodesic flows on manifolds of strictly negative curvature enjoy exponential decay of correlations \cite{Li1, Do1} we are in a setting very close to the one in \cite{Do2}, i.e. we have a slow-fast system in which the fast variables have strong mixing properties. 

Unfortunately, the perturbation to the geodesic flows in \eqref{eq:full} it is not small when $e_x=\cO(\ve)$, so at low energies one is bound to lose control on the statistical properties of the dynamics. The only easy way out would be to prove that the limit system spends very little time in configurations in which one particle has low energy.\footnote{To investigate low energy situations directly for the coupled geodesic flows seems extremely hard: when the kinetic energy is comparable with the potential energy all kind of uncharted behaviors, including coexistence of positive entropy and elliptic islands, could occur!}  If this were the case, then one could first introduce a modified system in which one offsets the bad behavior at small energies and then tries to remove the cutoff by showing that, in the limit process, the probability to reach very small energies is small. We will pursue precisely such a strategy.

We now define the modified process. Since our equations are Hamiltonian with Hamiltonian $H=\sum_{x\in\Lambda} \frac 12 p_x^2+\frac{\ve}2\sum_{|x-y|=1}V(q_x,q_y)$, the simplest approach is to modify the kinetic part of the Hamiltonian making it homogeneous of degree one at low velocities and decreasing correspondingly the interaction at low energies. More precisely, given any two functions $\vf,\phi\in\cC^\infty(\bR_+\setminus\{0\},\bR)$,  consider the Hamiltonians
$H_{\vf,\phi}=\sum_{x\in\Lambda} \vf(e_x)+\frac{\ve}2\sum_{|x-y|=1}\phi(e_x)\phi(e_y)V(q_x,q_y)$, which yield the equations of motion
\[
\begin{split}
&\dot q_x=\vf'(e_x) \bi(p_x)+\ve\sum_{|x-y|=1}\phi'(e_x)\phi(e_y)V(q_x,q_y) \bi(p_x)\\
&\dot p_x=\vf'(e_x) \tilde F(q_x,p_x)+\ve\sum_{|x-y|=1}\phi'(e_x)\phi(e_y)V(q_x,q_y)\tilde F(q_x,p_x)\\
&\quad\quad\quad-\ve\sum_{|y-x|=1}\phi(e_x)\phi(e_y)\;d_{q_x}V\\
&\dot e_x=-\ve\sum_{|y-x|=1}\phi(e_x)\phi(e_y)\; d_{q_x}V(\bi( p_x)),
\end{split}
\]
with $\tilde F$ has in \eqref{eq:hami-geo}.\footnote{By $d_{q_x}V$ we mean the differential of the function $V(\cdot, q_y)$ for any fixed $q_y$.}
Which, in the variables $(q_x,v_x,e_x)$, reads
\begin{equation}\label{eq:fullm}
\begin{split}
&\dot q_x=\sqrt{2e_x}\vf'(e_x) v_x+\ve\sum_{|x-y|=1}\sqrt{2e_x}\phi'(e_x)\phi(e_y)V(q_x,q_y) v_x\\
&\dot v_x=\vf'(e_x) \sqrt{2e_x}F(q_x,v_x)+\ve\bigg\{\sum_{|x-y|=1}\phi'(e_x)\phi(e_y)\sqrt{2e_x}V(q_x,q_y)F(q_x,v_x)\\
&\quad\quad-\sum_{|y-x|=1}\frac{\phi(e_x)\phi(e_y)}{\sqrt{2e_x}}\nabla_{q_x}V
+\sum_{|y-x|=1}v_x\frac{\phi(e_x)\phi(e_y)}{\sqrt{2e_x}}\;d_{q_x}V(v_x)\bigg\}\\
&\dot e_x=-\sum_{|y-x|=1}\phi(e_x)\phi(e_y)\sqrt{2e_x}\;d_{q_x}V(v_x),
\end{split}
\end{equation}
with $F$ as in \eqref{eq:indip-geo}.

Since $\frac{d}{dt}v_x^2=\ve (v_x^2-1)\sum_{|y-x|=1}\frac{\phi(e_x)\phi(e_y)}{\sqrt{2e_x}}d_{q_x}V(v_x)$, the manifold $v_x^2=1$ is an invariant manifold for the equations \eqref{eq:fullm}, thus such equations determine a flow in the variables $(\xi_x, e_x)=(q_x,v_x, e_x)\in T^1M\times \bR_+$.

Finally, we chose $\vf=\vf_\delta$ and $\phi=\phi_\delta$ such that, for all $\delta>0$,
\begin{equation}\label{eq:fifi}
\vf_\delta(s)=\begin{cases}  s \quad&\text{ if } s\geq \delta\\
                                 2\sqrt{\delta s}&\text{ if } s\leq \frac \delta {8}
\end{cases}
\quad;\quad
\phi_\delta(s)=\frac{1}{\vf_\delta'(s)}=\begin{cases}  1 \quad&\text{ if } s\geq \delta\\
                                 \frac{\sqrt{s}}{\sqrt \delta}&\text{ if } s\leq \frac \delta {8},
\end{cases}
\end{equation}
where $\phi_\delta$ is increasing.

We denote the solution of the above equations \eqref{eq:fullm} with initial conditions $(\xi,e)$
by $(\xi^{\ve,\delta}(t), e^{\ve,\delta}(t)).$ 

Our goal is to apply \cite[Theorem 7]{Do2} to the 
flow $(\xi^{\ve,\delta}(t), e^{\ve,\delta}(t))$, 
see Appendix \ref{sec:dolgopyat} for a simplified 
statement (Theorem \ref{thm:dolgopyat}) adapted to our needs.  
Before discussing the applicability of this a Theorem, there is one 
last issue we need to take care of: the equation for $e$ is clearly degenerate at low energies, this is related to the fact that the energies in \eqref{eq:fullm} are strictly positive for all times if they are strictly positive at time zero.\footnote{Indeed, the equation for the energy can be written, near zero, as $\dot e_x=-\ve e_x G(e_{\neq x},\xi)$, where $G$ is a bounded function, hence the solution has the form $e_x(t)=e_x(0) e^{-\ve \int_0^tG(e_{\neq x}(s),\xi(s))ds}$.} This may create a problem in the limiting process that is bound to have a degenerate diffusion coefficient.  To handle this problem it turns out to be much more convenient to use the variables $z_x=\ln e_x$.
In this new variables we finally have the equations we are looking for
\begin{equation}\label{eq:fullm2}
\begin{split}
&\dot q_x=\omega_\delta(z_x) v_x+\frac{\ve}2\sum_{|x-y|=1}\zeta_\delta(z_x)\phi_\delta(e^{z_y})V(q_x,q_y) v_x\\
&\dot v_x=\omega_\delta(z_x)F(q_x,v_x)+\frac{\ve}2\sum_{|x-y|=1}\zeta_\delta(z_x)\phi_\delta(e^{z_y})V(q_x,q_y)F(q_x,v_x)\\
&\quad\quad-\frac{\ve}{\sqrt 2}\sum_{|y-x|=1}e^{-\frac{z_x}2}\phi_\delta(e^{z_x})\phi_\delta(e^{z_y})\nabla_{q_x}V(q_x,q_y)\\
&\quad\quad+\frac{\ve}{\sqrt 2}\sum_{|y-x|=1}v_xe^{-\frac{z_x}2}\phi_\delta(e^{z_x})\phi_\delta(e^{z_y})L_{x}V(\xi_x,\xi_y)\\
&\dot z_x=-\ve\sqrt 2\sum_{|y-x|=1}e^{-\frac{z_x}2}\phi_\delta(e^{z_x})\phi_\delta(e^{z_y})L_{x}V(\xi_x,\xi_y),
\end{split}
\end{equation}
Where $L_x$ is as in equation \eqref{GenGF} and 
\begin{equation}\label{eq:cut-off}
\begin{split}
&\omega_\delta(z)=\sqrt{2}e^{\frac{z}2}\vf_\delta'(e^{z})=\begin{cases} \sqrt{2}e^{\frac{z}2}\quad& \text{ if } z\geq \ln\delta\\
                                                                                                \sqrt{2 \delta}& \text{ if } z\leq \ln\delta-\ln 8
                                                                           \end{cases}\\
&\zeta_\delta(z)=\sqrt{2}e^{\frac{z}2}\phi_\delta'(e^{z})=\begin{cases} 0\quad& \text{ if } z\geq \ln\delta\\
                                                                                               \frac{1}{\sqrt{2\delta}}& \text{ if } z\leq \ln\delta-\ln 8
                                                                           \end{cases}                                                                                                                                             
\end{split}
\end{equation}

\begin{rem}\label{rem:vector}
Note that we can chose $\omega_\delta\geq \sqrt \delta$ and $\zeta_\delta\geq 0$ decreasing.\footnote{Indeed,
\[
\phi_\delta(s)=1-\int_{\min\{s,\delta\}}^{\delta}\;\frac{\zeta_\delta(\ln x)}{\sqrt {2x}}\,dx\,.
\]
Remark that once $\zeta_\delta$ is chosen all the functions are fixed.}
In addition, it is possible to arrange that $|\omega_\delta|_{\cC^r(I_L,\bR)}\leq C_r e^L$, where $I_L=(-\infty, 2L)$, and  $|\zeta_\delta|_{\cC^r(\bR,\bR)}\leq C_r\delta^{-\frac 12}$, for each $r\in\bN, L, \delta\in\bR_+$. We will assume such properties in the following.
\end{rem}
Since the total energy is conserved, we can consider equations \eqref{eq:fullm2} on the set $(T^1M)^{\Lambda}\times(-\infty, L]^{\Lambda}$ for some $L>0$. Hence, by the above remark together with the \eqref{eq:fifi}, the vector field in \eqref{eq:fullm2} has bounded $\cC^r$ norm, as a function of $x,z,\ve$, for each $r\in\bN$.

Let  $\tilde f^\delta(\xi,z,\ve,\delta)=\xi^{\ve,\delta}(1)$, 
$F_{\ve,\delta}(\xi,z)=(\xi^{\ve,\delta}(1),z^{\ve,\delta}(1))$, and
\begin{equation}\label{eq:adelta}
A^\delta_x(\xi,z,\ve)=-\sqrt 2\!\int_0^1\!\!\!\sum_{|x-y|=1}\!\!\!\!\!e^{-\frac {z_x(\tau)}2}\phi_\delta(e^{z_x(\tau)})\phi_\delta(e^{z_y(\tau)})L_xV(\xi^{\ve,\delta}_x(\tau), \xi^{\ve,\delta}_y(\tau)) d\tau
\end{equation}
 then
\begin{equation}\label{eq:atlast}
F_{\ve,\delta}(\xi,z)=\left(\tilde f^\delta(\xi,z,\ve), z+\ve A^\delta(\xi,z,\ve)\right).
\end{equation}
\begin{lem}\label{lem:dynamics-prop}
Setting $\tilde F_\delta(x,z,\ve)=F_{\ve,\delta}(x,z)$ we have, for each $\delta\in(0,1), L>0$,  $\tilde F_\delta\in \cC^\infty((T^1M)^{\Lambda}\times (-\infty,L]^{\Lambda}\times [0,1])$, and $\|A^\delta(\cdot,\cdot,\ve)\|_{\cC^r((T^1M)^{\Lambda}\times (-\infty,L]^{\Lambda})}\leq C_{r,\delta}$, for each $r\in\bN, \ve\in[0,1]$. In addition, for each $\beta\in\bR_+$, the probability measure 
\[
\begin{split}
&d\mu_{\delta,\ve,\beta}=\tilde Z_\beta^{-1} e^{-\beta\tilde H_{\delta,\ve}+\sum_x\frac{d}2z_x}dqdv dz\\
& \tilde H_{\delta,\ve}(q,\nu,z)=\sum_{x\in\Lambda}\vf_\delta(e^{z_x})+\frac{\ve}2\sum_{|x-y|=1}\phi_\delta(e^{z_x})\phi(e^{z_y})V(q_x,q_y),
\end{split}
\]
is invariant for $F_{\ve,\delta}$. Moreover, for each $\bar z\in\bR^d$ and sub-manifold $\Sigma_{\bar z}:=\{z_x=\bar z_x\}$, the Dynamical System $(\Sigma_{\bar z},F_{0,\delta})$ has a unique SRB measure $\mu_{\bar z}$.
\end{lem}
\begin{proof}
The first part of the statement follows from Remark \ref{rem:vector} and subsequent comments together with standard results of existence of solutions and smooth dependence on the initial data from O.D.E.. The bound on $A^\delta$ is then immediate from formula \eqref{eq:adelta}.

By the Hamiltonian nature of the equations \eqref{eq:fullm} the measures 
\[
d\mu_{\delta,\beta}=Z_\beta^{-1} e^{-\beta H_{\vf_\delta,\phi_\delta}}dqdp\,,
\]
are invariant for the associated dynamics for each $\beta>0$. By changing variables we obtain the statement of the Lemma. 

Finally, calling $\tilde \mu$ the Riemannian measure on $T^1M$ we have that $\mu_{\bar z}=\otimes^{|\Lambda|}\tilde \mu$ is a SRB measure for the map $\xi \mapsto \tilde f^\delta(\xi,z,0)$, which turn out to be the product of the time $\omega_\delta(z_x)$ maps of the geodesic flow on $T^1M$. The uniqueness of the SRB follows by the mixing of the geodesic flows \cite{AS} and the fact that the product of mixing systems is mixing.
\end{proof}

\section{Existence of the limit: $\delta>0$}\label{sec:existence}
We are finally ready to consider the limit $\ve\to 0$, for the modified dynamics.

\begin{prop}\label{lem:itapplies}
For each $\delta\in (0,1)$ there exists $\ve_\delta>0$ such that the Dynamical System defined by \eqref{eq:atlast} satisfies the hypotheses of Theorem \ref{thm:dolgopyat} for $\ve\in[0,\ve_\delta]$. 

Hence, the family $z^{\ve,\delta}(\ve^{-2}t)$ is tight and its weak accumulation points are a solution of the Martingale problem associated to the stochastic differential equation
\begin{equation}\label{eq:limit}
\begin{split}
d z_x^\delta&=a^\delta_x(z^\delta)dt+\sum_{y}\sigma_{xy}^\delta(z^\delta) dB_y,\\
z_x^\delta(0)&=\bar z_x
\end{split}
\end{equation}
where
\begin{equation}\label{eq:variance}
\begin{split}
&(\sigma^\delta)^2_{xy}(z)=\sum_{n\in\bZ}\int_{(T^1M)^\Lambda}A^\delta_x((\tilde f^\delta)^n(\xi, z, 0)A^\delta_y(\xi,z,0) d\mu_z\\
&=2\int_{-\infty}^{+\infty}\!\!\!\!\! dt\sum_{\substack{ |x-w|=1\\ |y-w'|=1}}\frac{\phi_\delta(e^{z_x})\phi_\delta(e^{z_w})\phi_\delta(e^{z_y})\phi_\delta(e^{z_{w'}})}{e^{\frac{z_x+z_y}2}}\\
&\quad\quad\quad\quad\times\bE\left(L_xV(\xi_x^{0,\delta}(t), \xi_w^{0,\delta}(t))\cdot L_yV(\xi_y, \xi_{w'})\right).
\end{split}
\end{equation}
Here $\bE$ is the expectation with respect to $\mu_{z}$ and $\|a^\delta\|_{\cC^0}+\|(\sigma^\delta)^2\|_{\cC^1}<\infty$. 
\end{prop}
\begin{proof}
First of all notice that the 
hypotheses on the smoothness of $F_{\ve,\delta}$ 
and the boundedness of $A^\delta$ are insured by Lemma \ref{lem:dynamics-prop}. 
Next, notice that $F_{0,\delta}(\xi,z)=(f_z^\delta(\xi), z)$ 
with $f_z^\delta(\xi)_x=g^{\omega_\delta(z_x)}(\xi_x)$, 
where $g^t$ is the geodesic flow on the unit tangent bundle $T^1M$, 
thus the $f_z^\delta$ are FAE.\footnote{FAEs are defined in 
Appendix \ref{sec:dolgopyat}. 
In our case, the abelian action is the one determined by 
the geodesic flows themselves, $\times_{i\in\Lambda}g^{t_i}.$
}  

Also we have that $\mu_{z}(A^\delta(\cdot,  z, 0))=0$. 
This follows by considering the transformation $\Theta(q,v)=(q, -v)$. 
Indeed $\Theta_*\mu_{z}=\mu_{z}$ while, 
the flow $\Psi^t_{\delta,\ve}$ associated to \eqref{eq:fullm2} 
satisfies 
$\Psi^t_{\delta,\ve}\circ \Theta= \Theta\circ \Psi^{-t}_{\delta,\ve}$. 
On the other hand, using the antisymmetry of $L_xV$ with respect to $v_x$,
\[
\begin{split}
A_x^\delta(\Theta(\xi), z,0)=&-\sqrt 2\int_0^1\sum_{|x-y|=1}e^{-\frac {z_x}2}\phi_\delta(e^{z_x})\phi_\delta(e^{z_y})L_xV\circ\Psi_{\delta,0}^\tau\circ\Theta(\xi) d\tau\\
=&\sqrt 2\int_0^1\sum_{|x-y|=1}e^{-\frac {z_x}2}\phi_\delta(e^{z_x})\phi_\delta(e^{z_y})L_xV\circ\Psi_{\delta,0}^{-\tau}(\xi) d\tau\\
=&-A_x^\delta(\Psi_{\delta,0}^{-1}(\xi), z,0).
\end{split}
\]
Thus $\mu_{z}(A^\delta(\cdot,  z, 0))=\mu_{z}(A^\delta(\Theta(\cdot),  z, 0))=-\mu_{z}(A^\delta(\Psi_{\delta,0}^{-1}(\cdot),  z, 0))=-\mu_{z}(A^\delta(\cdot,  z, 0))$, by the invariance of the measure.

The last thing to check is the uniform decay of correlation. Since $\omega_\delta\geq \sqrt\delta$, the results in \cite{Do1,Li1} imply\footnote{\cite{Do1} 
proves the exponential decay of correlations 
for geodesic flows on negatively curved surfaces,
\cite{Li1} extends the results to any negatively curved manifold.} 
that the $f_z$ are FAE with uniform exponential decay of correlation. In fact, in Theorem \ref{thm:dolgopyat} the decay of correlations is meant in a very precise technical sense. To see that the results in \cite{Li1} imply the wanted decay we must translate them into the language of standard pairs in which it is formulated Theorem \ref{thm:dolgopyat}. 
Let us start by stating the result in \cite{Li1}: let $g^a$ be the time $a$ map of the geodesic flow on the unit tangent bundle. For each smooth function $A$ let $\|A\|_s=\|A\|_\infty+\|\partial^s A\|_\infty$ where $\partial^s$ is the derivative in the weak stable direction. Then there exists $C,c>0$ such that, for each $z$ and $\rho, A\in\cC^1$,  holds true
\begin{equation}\label{eq:flow-dec}
\left|\bE(\rho\cdot A\circ \tilde g^{an})-\bE(A)\bE(\rho)\right|\leq C\|\rho\|_{\cC^1}\|A\|_s e^{-c an}.
\end{equation}
Since, setting $f^\delta_z(\xi)=\tilde f^\delta(\xi,z,0)$, $f^\delta_z=\times_{x} g^{\omega_\delta(z_x)}$, and $\omega_\delta$ is uniformly bounded from below, for $\bE(A)=0$, it follows (suppressing, to ease notation, the superscript $\delta$)\footnote{Just note that one can write
$\bE(\rho\cdot A\circ f_z^n)=\bE(\bE(\rho\cdot A\circ f_z^n\;|\; \xi_{y\neq x}))$ and that the relevant norms of $\rho_{\xi_{y\neq x}}(\xi_x)=\rho(\xi_x, \xi_{y\neq x})$ and $A_{\xi_{y\neq x}}(\xi_x)=A(\xi_x , f_z^n(\xi_{y\neq x}))$ are bounded by the full norms of $\rho$ and $A$. One can then apply \eqref{eq:flow-dec} to $\bE(\rho\cdot A\circ f_z^n\;|\; \xi_{y\neq x})=\bE(\rho_{\xi_{y\neq x}}A_{\xi_{y\neq x}}\circ \tilde f_{\omega_\delta(z_x)}^n)$. Proceeding in such a way one variable at a time yields the result.}
\begin{equation}\label{eq:flow-prod}
\left|\bE(\rho\cdot A\circ f_z^n)\right|\leq C|\Lambda|\,\|\rho\|_{\cC^1}\|A\|_s e^{-can}.
\end{equation}
To see that this is stronger than needed, consider a standard pair $\ell=(D,\rho)$.\footnote{Recall that $D$ is a manifold of fixed size close to the strong unstable one and $\rho$ a smooth density on it.} One can smoothly foliate a $\ve$ neighborhood of $D$ and define a probability density $\rho_\ve$ supported in it such that $\|\rho_\ve\|_{\cC^1}\leq C\ve^{-2}$, while $\|\rho_\ve\|_{\cC^1}\leq C$ when $\rho_\ve$ is restricted to a leaf of the foliation.
Thanks to the $\alpha$-H\"older regularity and the absolute continuity of the weak stable foliation, one can take $\rho_\ve$ so that
\[
\left|\bE_\ell(A)-\bE(\rho_\ve A)\right|\leq C\ve^\alpha \|A\|_{s}.
\]
Accordingly,
\[
\begin{split}
\left|\bE_\ell( A\circ f_z^n)\right|&\leq \left|\bE(\rho_\ve\cdot A\circ f_z^n)\right|+C\ve^\alpha\|A\circ f_z^n\|_s
\leq C\left\{\ve^{-2}e^{-can}+\ve^\alpha \right\}\|A\|_{s}\\
&\leq Ce^{-\frac {\alpha c an}{2+\alpha}}\|A\|_{\cC^1},
\end{split}
\]
where, in the last equality, we have chosen $\ve=e^{-\frac {c an}{2+\alpha}}$. 
Thus, all the hypotheses of Theorem \ref{thm:dolgopyat} are satisfied and \eqref{eq:variance} follows by a direct computation.
\end{proof}

By Theorem \ref{thm:dolgopyat}(b), in order to 
prove that $z^{\ve,\delta}(\ve^{-2}t)$ has a limit it 
suffices to prove that \eqref{eq:limit} has a unique solution. 
This would follow by standard results if we knew that $a^\delta$ 
is locally Lipschitz. In fact, \cite{Do2} provides also an explicit formula 
for $a^\delta.$ Unfortunately this formula is much more complex than 
the formula for the variance and is quite difficult to investigate. 
We will avoid a direct computation of $a_\delta$ and we will instead use the knowledge of the invariant measure to determine it. Before doing that a deeper understanding of the variance is required.

\section{Computing the variance}\label{sec:compute}
Let $g^t$ be the geodesic flow on the unit cotangent bundle of $M$. As already noted, for each function $h$, $h(\xi_x^{0,\delta}(t))=h\circ g^{\omega_\delta(z_x)t}(\xi_x)$ for all $x\in\Lambda$. For convenience let us set $\varpi_x:=\omega_\delta(z_x)$. Also, it turns out to be useful to define two functions of two variables: consider two geodesic flows on $T^1M$, let $(\xi,\eta)$ be the variables of the two flows respectively,  $\bE$ the expectation with respect to the Riemannian volume on $(T^1M)^2$ and $L_1, L_2$ the generators associated to the geodesic flow of $\xi$ and $\eta$ respectively, then we define $\rho,\tilde\rho:\bR^2\to\bR$ by
\begin{equation}\label{eq:block}
\begin{split}
\rho(a,b)&:=\int_{-\infty}^\infty dt\; \bE\left(L_1V(g^{at}(\xi), g^{bt}(\eta))\cdot L_{1}V(\xi, \eta)\right), \\
\tilde\rho(a,b)&:=\int_{-\infty}^\infty dt\; \bE\left(L_1V(g^{at}(\xi), g^{bt}(\eta))\cdot L_2V(\xi, \eta)\right).
\end{split}
\end{equation}
Also, it is convenient to define 
\begin{equation}\label{eq:rhoxy}
\rho_{xy}:=\rho(\omega_\delta(z_x),\omega_\delta(z_y)), \quad \tilde \rho_{xy}:=\tilde\rho(\omega_\delta(z_x),\omega_\delta(z_y)).
\end{equation}

Indeed, the understanding of the variance will be reduced shortly  to understanding the properties of $\rho_{xy}$. Here is a list of relevant properties whose proof can be found in Appendix \ref{app:rho}.
\begin{lem}\label{lem:rho-tilde}
The function $\tilde \rho$ is non-positive and $\cC^\infty$ for $a,b>0$.
In addition, for each $a,b, \lambda>0$ we have $\tilde\rho(a,b)=\tilde\rho(b,a)$ and  $\rho(\lambda a,\lambda b)=\lambda^{-1}\rho(a,b)$.
Finally, $\tilde\rho(a,b)=-\frac{a}{b}\rho(a,b)$.
\end{lem}

\begin{rem} Note that the previous Lemma implies $a^2\rho(a,b)=b^2\rho(b,a)$.
\end{rem}

\begin{lem} \label{lem:rho-asymp} 
There exists $A,B>0$ such that, for all $a, b>0$,
\[
\left|\rho(a,b)-\frac {A\,b^2}{a^3+b^3}\right| \leq \frac {B\,ab^3}{a^5+b^5}.
\]
Finally, for all $a,b>0$,
\[
\left|\partial_a\rho(a,b)\right|\leq \frac{B\,ab^2}{a^5+b^5}\;;\quad \quad a\partial_a\rho(a,b)+b\partial_b\rho(a,b)=-\rho(a,b).
\]
\end{lem}
We are now in the position to derive an helpful formula for the variance.
\begin{lem}\label{lem:sigma-formula}
The following formula holds true
\[
(\sigma^\delta)^2_{xy}(z)=\begin{cases}2e^{-z_x}\sum_{|x-w|=1}\{\phi_\delta(e^{z_x})\phi_\delta(e^{z_w})\}^2\rho_{xw}&\quad\text{ if }x=y\\
                                              -2 e^{-z_y}\phi_\delta(e^{z_x})\phi_\delta(e^{z_y})^3\rho_{xy}&\quad\text{ if }|x-y|=1\\
                                              0&\quad\text{ if }|x-y|>1.\\
                         \end{cases}
\]
\end{lem}
\begin{proof}
Remembering \eqref{eq:variance}, given any two couples of neighboring sites $x,w$, $y,w'$ we want to compute
\[
\int_{-\infty}^\infty dt \;\bE\left(L_xV(g^{\varpi_xt}(\xi_x), g^{\varpi_wt}(\xi_w))\cdot L_yV(\xi_y, \xi_{w'})\right).
\]
In fact, remembering the properties of the transformation $\Theta$ in the proof of Lemma \ref{lem:itapplies}, it suffices to compute the integral on $[0,\infty)$.

Since $\bE(v_x\;|\; q_{\neq x}, v_{\neq x})=0$, it follows that the above integral is different from zero only if
$x=y$ or $x=w'$ and $w=y$. On the other hand if $x=y$, since $g^{at}\times g^{bt}$ is a mixing flow for each $a,b>0$, we can write
\[
\begin{split}
&\int_0^\infty dt\; \bE\left(\left\{\varpi_x^{-1}\frac d{dt}V(g^{\varpi_xt}(\xi_x), g^{\varpi_wt}(\xi_w))-\frac {\varpi_w}{\varpi_x}L_wV(g^{\varpi_xt}(\xi_x), g^{\varpi_wt}(\xi_w))\right\}\cdot L_xV(\xi_x, \xi_{w'})\right)\\
&= \varpi_x^{-1} \bE\left(V(q_x,q_w)\right)\bE\left(L_xV(\xi_x, \xi_{w'})\right)- \varpi_x^{-1}\bE\left(V(q_x, q_w)\cdot L_xV(\xi_x, \xi_{w'}) \right)\\
&\quad-\frac{\varpi_w}{\varpi_x}\int_0^\infty dt\;\bE\left(L_wV(g^{\varpi_xt}(\xi_x), g^{\varpi_wt}(\xi_w))\cdot L_xV(\xi_x, \xi_{w'})\right)\\
&=-\delta_{w, w'}\frac{\varpi_w}{\varpi_x}\int_0^\infty dt\;\bE\left(L_wV(g^{\varpi_xt}(\xi_x), g^{\varpi_wt}(\xi_w))\cdot L_xV(\xi_x, \xi_{w})\right)\\
&=\delta_{w, w'}\int_0^\infty dt\;\bE\left(L_xV(g^{\varpi_xt}(\xi_x), g^{\varpi_wt}(\xi_w))\cdot L_xV(\xi_x, \xi_{w})\right).
\end{split}
\]
Thus, remembering \eqref{eq:fifi}, \eqref{eq:cut-off} and that $\varpi_x=\omega_\delta(z_x)$,
\[
\sigma^2_{xx}=2e^{-z_x} \sum_{|x-w|=1}\phi_\delta(e^{z_x})^2\phi_\delta(e^{z_w})^2\rho_{xw},
\]
and $\sigma^2_{xy}=0$ if $|x-y|>1$. If $|x-y|=1$, then (remembering the symmetry of the potential and using Lemma \ref{lem:rho-tilde})
\[
%\begin{split}
\sigma^2_{xy}=2\phi_\delta(e^{z_x})^2\phi_\delta(e^{z_y})^2 e^{-\frac{z_x+z_y}2}\tilde\rho_{xy}%\\
=-2e^{-z_y}\phi_\delta(e^{z_x})\phi_\delta(e^{z_y})^3\rho_{xy}.
%\end{split}
\]
\end{proof}

\section{The limit equation ($\delta>0$): structure}\label{sec:structure}
Having gained a good knowledge on the variance we are ready to write the limit equation in a more explicit and convenient form.

We introduce standard Brownian motions $B_{xy}$ indexed by oriented edges,
so that the motions associated to different non oriented edges are 
independent and $B_{xy}=-B_{yx}.$ 
Considering the Gaussian processes $W_x:=\sum_{|x-y|=1}\beta_{xy}(z)B_{xy}$ we have
\[
\bE( W_x(t)W_y(t)\;|\;z)=\begin{cases}\sum_{|x-w|=1}\beta_{xw}(z)^2\; t\quad&\text{ for } x=y\\
                           -\beta_{xy}(z)\beta_{yx}(z)\;t\quad&\text{ for } |x-y|=1\\
                           0\quad&\text{ for } |x-y|>1.
                           \end{cases}
\]
We set\footnote{This is well defined since $\rho_{x,y}\geq 0$ by Lemma \ref{lem:rho-tilde}.} 
\begin{equation}\label{eq:beta-sde}
\beta_{xy}(z)=\sqrt 2 e^{-\frac{z_x}2}\phi_\delta(e^{z_x})\phi_\delta(e^{z_y})\sqrt{\rho_{xy}}\,,
\end{equation}
hence, remembering Lemmata \ref{lem:sigma-formula}, \ref{lem:rho-tilde} and equations \eqref{eq:rhoxy},\eqref{eq:cut-off},\eqref{eq:fifi},
\[
(\sigma^\delta)^2_{xy}(z)=\begin{cases}\sum_{|x-w|=1}\beta_{xw}^2&\quad\text{ if }x=y\\
                                              -\beta_{xy}\beta_{yx}&\quad\text{ if }|x-y|=1\\
                                              0&\quad\text{ if }|x-y|>1.\\
                         \end{cases}
\]
Then, we can write \eqref{eq:limit} as
\begin{equation}\label{eq:limit-e}
dz_x^\delta=a^\delta_x(z^\delta)dt+\sum_{|x-y|=1}\beta_{xy}(z^\delta)\;dB_{xy}.
\end{equation}

Let $\cL$ be the operator in the Martingale problem associated to the diffusion defined by \eqref{eq:limit}.

\begin{lem} \label{lem:reversi} 
If the manifold $M$ is $d$ dimensional, then for each $\beta>0$, 
\[
e^{\sum_x\frac{ d}2z_x-\beta \vf_\delta(e^{z_x})}dz
\]
is an invariant measure for the process defined by \eqref{eq:limit-e}. In addition, the process \eqref{eq:limit-e} is reversible. That is, calling $\bE_\beta$ the expectation with respect to the above invariant measure,
\[
\bE_\beta(\vf\cL h)=\bE_\beta(h\cL \vf)
\]
for each smooth real functions $\vf, h$.
\end{lem}
\begin{proof}
Recall that Lemma \ref{lem:dynamics-prop} gives the invariant measures of the original Dynamical System. 
In particular , for each $\psi\in\cC^0(\bR^{|\Lambda|},\bR)$
\[
|\mu_{\delta,\ve,\beta}(\psi(z^{\ve,\delta}(\ve^{-2}t)))-\mu_{\delta,0,\beta}(\psi(z^{\ve,\delta}(\ve^{-2}t)))|\leq C\ve |\psi|_\infty
\]
Thus
\[
|\mu_{\delta,0,\beta}(\psi(z^{\ve,\delta}(\ve^{-2}t)))-\mu_{\delta,0,\beta}(\psi(z^{\ve,\delta}(0)))|\leq 2C\ve|\psi|_{\infty}.
\]
Taking the limit $\ve\to 0$ along any subsequence leading to an accumulation point we see that $\mu_{\delta,0,\beta}$ is an invariant measure for the process \eqref{eq:limit}. The claim of the Lemma now follows by taking the marginal of $\mu_{\delta,0,\beta}$ in the variables $z$.

In the same manner, using the same notation as in the proof of Lemma \ref{lem:itapplies}, for each continuos functions $\psi, g$ and converging sequence $z_{\ve_k,\delta}(\ve_k^{-2}t)$ we have
\[
\begin{split}
\bE_\beta(\psi(z(t)) g(z))&=\lim_{k\to\infty}\mu_{\delta,\ve_k,\beta}(g\cdot\psi\circ \Psi_{\ve_k,\delta}^{\ve_k^{-2}t} )
=\lim_{k\to\infty}\mu_{\delta,\ve_k,\beta}(\psi \cdot g\circ \Psi_{\ve_k,\delta}^{-\ve_k^{-2}t} )\\
&=\lim_{k\to\infty}\mu_{\delta,\ve_k,\beta}(\psi\circ \Theta \cdot g\circ\Theta\circ \Psi_{\ve_k,\delta}^{\ve_k^{-2}t} )
=\bE_\beta(g\circ \Theta(z(t)) \psi\circ \Theta(z)).
\end{split}
\]
Since $g,\psi$ are functions of the $z$ only, it follows $g\circ \Theta=g$, $\psi\circ \Theta=\psi$ and
\[
\bE_\beta(\psi(z(t)) g(z))=\bE_\beta(g(z(t)) \psi(z)).
\]
Differentiating with respect to $t$ at $t=0$ yields the Lemma.
\end{proof}

\begin{lem}\label{lem:drift-formula} The drift $a^\delta_x$ has the form 
\[
\begin{split}
a^\delta_x=&\sum_{|x-y|=1}\left\{\partial_{z_x}\left[e^{-z_x}\phi_\delta(e^{z_x})^2\phi_\delta(e^{z_y})^2\rho_{xy}\right]-
\partial_{z_y}\left[e^{-z_y}\phi_\delta(e^{z_x})\phi_\delta(e^{z_y})^3\rho_{xy}\right]\right\}\\
&+\frac{d}{2}\sum_{|x-y|=1}\left[e^{-z_x}\phi_\delta(e^{z_x})^2\phi_\delta(e^{z_y})^2-e^{-z_y}\phi_\delta(e^{z_x})\phi_\delta(e^{z_y})^3\right]\rho_{xy}.
\end{split}
\]
\end{lem}
\begin{proof}
The idea to compute the $a^\delta_x$ is very simple: first compute $\cL$ and $\cL^*$ and then check what the reversibility condition implies. The operator associated to the diffusion \eqref{eq:limit} is given by 
\[
\cL=\sum_x a^\delta_x\partial_{z_x}+\frac 12\sum_{x,y}(\sigma^\delta)^2_{xy}\partial_{z_x}\partial_{z_y}.
\]
The adjoint $\cL^*$ with respect to the invariant measures in Lemma \ref{lem:reversi} can then 
be computed by integrating by parts.
Setting $\Gamma_x(z):=\frac {d}2-\beta \phi_\delta(e^{z_x})^{-1}$ we have
\[
\begin{split}
\cL^* \psi&=-\sum_x \{\partial_{z_x}a^\delta_x+ a^\delta_x\Gamma_x\}\psi-\sum_x a^\delta_x\partial_{z_x}\psi\\
&+\frac 12\sum_{xy}\left[\partial_{z_x}\partial_{z_y}(\sigma^\delta)^2_{xy}+2\Gamma_x\partial_{z_y}(\sigma^\delta)^2_{xy}
+\Gamma_x\Gamma_y(\sigma^\delta)^2_{xy}+\delta_{xy}\partial_{z_x}\Gamma_x(\sigma^\delta)^2_{xy}\right] \psi\\
&+\sum_{xy}\left[\partial_{z_y}(\sigma^\delta)^2_{xy}+\Gamma_y (\sigma^\delta)^2_{xy}\right]\partial_{z_x}\psi
+\frac 12\sum_{xy} (\sigma^\delta)^2_{xy}\partial_{z_x}\partial_{z_y}\psi.
\end{split}
\]
This implies
\[
a^\delta_x=\frac 12\sum_y\left[\partial_{z_y}(\sigma^\delta)^2_{xy}+\Gamma_y(\sigma^\delta)^2_{xy}\right]
\]
and the Lemma follows by direct algebraic computations using Lemma \ref{lem:sigma-formula}.
\end{proof}

The next result is an obvious fact that is nevertheless of great importance.
\begin{lem}\label{lem:energy}
The function $\cH:=\sum_x\vf_\delta(e^{z_x})$ is constant in time.
\end{lem}
\begin{proof}
It is useful to notice that, setting $\psi_x:=\frac{e^{z_x}}{\phi_\delta(e^{z_x})}$ and $\kappa_{xy}=\psi_x\beta_{xy}$,
$\kappa_{xy}=\kappa_{yx}$.

By Ito's formula we have
\[
d\cH=\sum_x\psi_x a_x dt+\sum_{|x-y|=1}\kappa_{xy}d B_{xy}+\frac 12\sum_x\partial_z\psi_x\sum_{|x-y|=1}\beta_{xy}^2dt.
\]
The second term is zero by the antisymmetry of $B_{xy}$, thus (using Lemma \ref{lem:drift-formula} and the symmetry of $\kappa_{xy}$ again)
\[
\begin{split}
d\cH=&\frac 12\sum_{|x-y|=1}\left[\psi_x\partial_{z_x}\beta_{xy}^2-\psi_y\partial_{z_x}\beta_{xy}\beta_{yx}\right]dt\\
&+\frac{d}2\sum_{|x-y|=1}\left[\psi_x^{-1}-\psi_y^{-1}\right]\kappa_{xy}^2 dt
+\frac 12\sum_{|x-y|=1}\beta_{xy}^2\partial_z\psi_x dt=0.
\end{split}
\]
\end{proof}
We conclude with the main result of this section.

\begin{prop}\label{prop:uniqueness} For each $\delta>0$ the family $z^{\ve,\delta}(\ve^{-2}t)$ converges weakly, for $\ve\to 0$, to the process $z(t)$ determined by the SDE  \eqref{eq:limit-e}.
\end{prop}
 \begin{proof}
 From Lemma \ref{lem:drift-formula} and  Lemma \ref{lem:rho-tilde} it follows that $a^\delta\in \cC^\infty$, this, together with the boundedness and convergence results established in Lemma \ref{lem:itapplies} and the standard results on the uniqueness of the solution of the SDE, imply that all the accumulation points of  $z^{\ve,\delta}(\ve^{-2}t)$ must coincide, hence the Proposition.
 \end{proof}

\section{The limit equation ($\delta=0$): properties and stopping times}\label{sec:reach}

It is natural to consider the {\em stopping time} $\tau_\delta:=\inf\{t\in\bR_+\;:\; \min_{x\in\Lambda} z_x\leq \ln \delta\}$.
In addition, Lemma \ref{lem:energy} suggests the convenience of going back to the more
{\em physical} process $\cE_x(t)=\vf_\delta(e^{z_x(t\wedge \tau_\delta)})=e^{z_x(t\wedge \tau_\delta)}$. 

\begin{lem} 
\label{CrDriftUp}
For each $t\leq \tau_\delta$, the process $\cE_x$ satisfies the SDE
\[
d\cE_x=\sum_y \ba(\cE_x,\cE_y) dt+\sqrt 2\bbeta(\cE_x,\cE_y) dB_{xy} 
\]
where $\ba,\bbeta\in\cC^\infty((0,\infty)^2,\bR)$ are respectively anti-symmetric and symmetric functions that satisfy \eqref{eq:drift-formula}, \eqref{eq:ab-asy}.
In addition, if $d\geq 3$, then for each constant 
\[
M\geq \max\left\{1,\frac{d-1+\frac{8B}A}{d-2}\right\},
\]
if $\cE_y> M \cE_x$, then 
\begin{equation}
\label{DriftLog}
\ba(\cE_x,\cE_y) \cE_x\geq \bbeta(\cE_x,\cE_y)^2.
\end{equation}
\end{lem}
\begin{proof}
By Ito's formula and \eqref{eq:limit-e} we have\footnote{Here we suppress the $\delta$-dependence since we stop the motion before seeing the region in which the dynamics has been modified.}
\begin{equation}\label{eq:ab-comp}
d\cE_x=\left[e^{z_x}a_x+\frac 12 e^{z_x}\sum_{|x-y|=1}\beta_{xy}^2\right]dt+\sum_{|x-y|=1}e^{z_x}\beta_{xy}dB_{xy}.
\end{equation}
Using \eqref{eq:beta-sde},\eqref{eq:rhoxy}, \eqref{eq:cut-off} and Lemma \ref{lem:rho-tilde} we can write
\begin{equation}\label{eq:beta-final}
e^{z_x}\beta_{xy}=\sqrt{2\cE_x\rho(\sqrt {2\cE_x},\sqrt {2\cE_y})}=:\sqrt 2\bbeta(\cE_x,\cE_y).
\end{equation}
Lemma \ref{lem:drift-formula}, equations \eqref{eq:fifi}, \eqref{eq:rhoxy} and \eqref{eq:cut-off} yield
\[
a_x=\sum_{|x-y|=1}\left[ \partial_{\cE_x}\rho-\partial_{\cE_y}\rho\right]+\frac{d-2}{2}\sum_{|x-y|=1}\left[\cE_x^{-1}-\cE_y^{-1}\right]\rho.
\]
Using  equation \eqref{eq:ab-comp} we finally obtain \eqref{eq:drift-formula} and from Lemma \ref{lem:rho-asymp} follows \eqref{eq:ab-asy}. 

Moreover, by Lemma \ref{lem:rho-asymp},
\[
\begin{split}
\partial_{\cE_x}\rho_{xy}&=\frac{1}{\sqrt{2\cE_x}}\partial_a\rho(\sqrt {2\cE_x},\sqrt {2\cE_y})\\
\partial_{\cE_y}\rho_{xy}&=-\frac{1}{2\cE_y}\left\{\rho(\sqrt {2\cE_x},\sqrt {2\cE_y})+\sqrt{2\cE_x}\partial_a\rho(\sqrt {2\cE_x},\sqrt {2\cE_y})\right\}\\
&=-\frac{\bbeta(\cE_x,\cE_y)^2}{2\cE_x\cE_y}-\frac{\cE_x}{\cE_y}\partial_{\cE_x}\rho_{xy}.
\end{split}
\]
Hence
\[
\begin{split}
\cE_x\ba(\cE_x,\cE_y)&=\bbeta^2+\cE_x^2\partial_{\cE_x}\rho_{xy}-\cE_x^2\partial_{\cE_y}\rho_{xy}+\frac{d-2}2\bbeta^2-\frac{d-2}2\cE_x\cE_y^{-1}\bbeta^2\\
&=\left\{\frac d2-\frac{d-1}2\frac{\cE_x}{\cE_y}\right\}\bbeta(\cE_x,\cE_y)^2+\cE_x^2\left\{1+\frac{\cE_x}{\cE_y}\right\}\partial_{\cE_x}\rho_{xy}.
\end{split}
\]
The regularity of the coefficients follows from the previous results and some 
algebraic computations. 
At last, for $\cE_y>M\cE_x$,
\[
\cE_x\ba(\cE_x,\cE_y)\geq\left\{\frac d2-\frac{d-1}{2M}\right\}\bbeta(\cE_x,\cE_y)^2-\frac{B(1+M^{-1})}{2M}\frac{\cE_x}{(2\cE_y)^{\frac 12}}.
\]
On the other hand
\begin{equation}\label{eq:diff-est}
\bbeta(\cE_x,\cE_y)^2=\cE_x\rho(\sqrt {2\cE_x},\sqrt {2\cE_y})\geq \cE_x\left[ \frac{A}{\sqrt{2\cE_y}}-\frac{B\sqrt{2\cE_x}}{2\cE_y}\right]\geq \frac{A\cE_x}{4\sqrt{2\cE_y}},
\end{equation}
from which the Lemma follows.
\end{proof}

\section{The limit equation ($\delta=0$): unreachability of zero energy}
\label{sec:zero}

Our last task it to prove that the stopping time $\tau_\delta$ tends to infinity when $\delta$ tends to zero or, in other words, energy zero is {\em unreachable} for the limit equation.

Fix any $T>0$.

For each subset $\Gamma\subset \Lambda$ let us define the energy of the cluster $\cE_\Gamma:=\sum_{x\in\Gamma}\cE_x$. Also, for each $\delta>0$, $n\in\{1,\dots,|\Lambda|\}$, let us define the stopping times 
\[
\tau_\delta^n:=\inf\{t\in[0,\infty)\;:\; \exists \;\Gamma\subset \Lambda, |\Gamma|=n, \cE_\Gamma(t)\leq \delta\}\wedge T.
\]
Note that $\tau^1_\delta=\tau_\delta\wedge T$, where $\tau_\delta$ is defined at the beginning of section \ref{sec:reach}.

\begin{lem}\label{lem:no-zero}
Let $\bP$ be the measure associated to the process \eqref{eq:final-eq}, then
\[
 \lim_{\delta\to 0}\bP\left( \left\{\tau^1_{\delta}<T\right\}\right)=0. 
 \]
\end{lem}
\begin{proof} We we will prove that for each $\eta>0$ and $n\in \{1,\dots|\Lambda|\}$ there exists $\delta_n=\delta_n(\eta)$, such that
\[
 \bP\left( \left\{\tau^n_{\delta_n}<T\right\}\right)\leq 2^{-n }\eta. 
 \]
The proof is by (backward) induction. The case $n=|\Lambda|$ follows by the energy conservation by choosing $\delta_{|\Lambda|}<\frac{\cE_{\Lambda}}2$.
 
Next, suppose the statement true for $n+1\leq  |\Lambda|$. 
It is convenient to define, for each $\Gamma\subset \Lambda$ the stopped process $\hat\cE_\Gamma(t)=\cE_\Gamma(t\wedge \tau^{n+1}_{\delta_{n+1}})$ and the set $\Omega=\{\tau_{\delta_{n+1}}^{n+1}\geq T\}$. Then, for each $0<\delta<\delta_{n+1}$, we have
\[
\begin{split}
\bP\left( \left\{\tau^n_{\delta}<T\right\}\right)&\leq \bP\left( \left\{\tau^n_{\delta}<T\right\}\cap \Omega\right) +2^{-(n+1)}\eta\\
&\leq \bP\left(\bigcup_{\substack{\Gamma\subset \Lambda\\|\Gamma|=n}}\left\{\inf_{t\in [0,T]} \hat\cE_\Gamma(t)\leq \delta\right\}\right)+2^{-(n+1)}\eta. 
\end{split}
\]

It thus suffices to show that there exists $\delta_n\leq \delta_{n+1}$ such that, for each $\Gamma\subset \Lambda$, $|\Gamma|=n$, we have
\[
\bP\left(\left\{\inf_{t\in [0,T]} \hat\cE_\Gamma(t)\leq \delta_n\right\}\right) 
\leq 2^{-(|\Lambda|+n+1)}\eta \leq \binom{|\Lambda|}{n}^{-1}2^{-(n+1)}\eta. 
\]
Let us fix $\Gamma\subset \Lambda$, $|\Gamma|=n$.

Observe that if $\Omega$ holds but $\cE_\Gamma(t) \leq \frac{\delta_{n+1}}{M+1}$ 
then $\cE_y\geq \frac{M\delta_{n+1}}{M+1}\geq M\cE_\Gamma\geq M\cE_x$ for all $y\not\in\Gamma$ and $x\in \Gamma$. In the following we will chose $M$ as in the statement of Lemma \ref{CrDriftUp}.

Next, we define the process $Y=\ln \cE_\Gamma$ which satisfies
\begin{equation}\label{eq:y-log}
dY= \sum_{(x,y)\in B(\Gamma)}\left\{\frac{ \ba( \cE_x, \cE_y) \cE_\Gamma-\bbeta( \cE_x, \cE_y)^2}{2 \cE_\Gamma^2} dt+
\sqrt 2\bbeta(\cE_x,\cE_y)\cE_\Gamma^{-1} dB_{xy}\right\},
\end{equation}
where $B(\Gamma)=\{(x,y)\in\Lambda^2\;:\; x\in \Gamma,\; y\not\in\Gamma,\;|x-y|=1\}$.

Observe that by Corollary \ref{CrDriftUp} the drift is positive, indeed
\[
\sum_{(x,y)\in B(\Gamma)}\frac{\left(\ba(\cE_x,\cE_y) \cE_\Gamma-\bbeta(\cE_x,\cE_y)^2\right)}{2 \cE_\Gamma^2} \geq
\sum_{(x,y)\in B(\Gamma)}\frac{\left( \ba(\cE_x,\cE_y) \cE_x-\bbeta(\cE_x,\cE_y)^2\right)}{2 \cE_\Gamma^2}\geq 0. 
\]
In addition, arguing as in \eqref{eq:diff-est}, if $\cE_\Gamma(t) \leq \frac{\delta_{n+1}}{M+1}$ we have, for some constant $C>0$,
\begin{equation}\label{eq:big-diff}
\cE_\Gamma^{-2}\bbeta(\cE_x,\cE_y)^2\leq 2\frac{\cE_x}{\cE_\Gamma^2}\left[ \frac{A}{\sqrt{2\cE_y}}+\frac{B\sqrt{2\cE_x}}{2\cE_y}\right]\leq  \frac{C}{\cE_\Gamma^\frac 32}.
\end{equation}
Therefore
\[
Y(t\wedge \tau^{n+1}_{\delta_{n+1}})\geq Y(0)+\int_0^{t\wedge \tau^{n+1}_{\delta_{n+1}}}\sum_{(x,y)\in B(\Gamma)}\sqrt 2\bbeta(\cE_x,\cE_y)\cE_\Gamma^{-1} dB_{xy}=:\bM(t).
\]
Note that $\bM$ is a Martingale. Let $\tau_*=\inf\{t\;:\; \bM(t)\leq \ln \delta_{n+1}\}\wedge T$. Consider the new martingale 
$\widetilde \bM(t)=\bM(t)-\bM(t\wedge \tau_*)$ and the stopping time 
$$\hat\tau=\inf\{t\;:\; \widetilde\bM(t)\leq\ln\delta_n-\ln\delta_{n+1}\text{ or }\widetilde\bM(t)\geq -\frac 12\ln\delta_{n+1}\}\wedge T.$$ 
Setting $p=\bP(\{\bM(\hat\tau)=\ln\delta_n\})$ we obtain
\[
0\leq p(\ln\delta_n-\ln\delta_{n+1})-(1-p)\frac 12\ln\delta_{n+1} ,
\]
which implies
\[
\bP(\{\bM(\hat\tau)=\ln\delta_n\})\leq \frac{\ln\delta_{n+1}}{2\ln\delta_n-\ln\delta_{n+1}}.
\]
Set $\delta_n=\delta_{n+1}^\alpha$, $\alpha>1$ to be chosen later. 
The probability that $\bM$, starting from $\ln\delta_{n+1}$ reaches $\ln\delta_n$ before reaching $\frac 12 \ln\delta_{n+1}$ is smaller than $(2\alpha-1)^{-1}$. Accordingly, the probability that the martingale reaches $\ln\delta_n$ before downcrossing $L$ times the interval $[\ln\delta_{n+1},\frac 12\ln\delta_{n+1}]$ is smaller than $1-(1- (2\alpha-1)^{-1})^L\leq \alpha^{-1}L$. On the other hand by Doob's inequality the expectation of the number of downcrossing is bounded by
$\frac 2{\ln\delta_{n+1}^{-1}}\bE((\bM-\frac 12\ln\delta_{n+1})^+)$. Since $\bM-\frac 12\ln\delta_{n+1}\geq 0$ implies
$\cE_\Gamma\geq\sqrt{\delta_{n+1}}$, by \eqref{eq:big-diff} follows 
\[
\bE((\bM-\frac 12\ln\delta_{n+1})^+)\leq C\delta_{n+1}^{-\frac 34},
\]
for some constant $C$ independent on $\ve$.
From this it immediately follows that the probability to have more than $L$ downcrossing is less that $L^{-1}\delta^{-1}_{n+1}$. In conclusion,  
\[
\bP\left(\left\{\inf_{t\in [0,T]} \hat\cE_\Gamma(t)\leq \delta_n\right\}\right)\leq C(\alpha^{-1}L+L^{-1}\delta_{n+1}^{-1})
\]
which yields the wanted estimate by first choosing $L^2=\alpha\delta_{n+1}^{-1}$ and then setting $\alpha=C^2\delta_{n+1}^{-1}2^{2|\Lambda|+2n+4}\eta^{-2}$.\footnote{Note that $\delta_n\sim\delta_{n+1}^{C\delta_{n+1}^{-1}}$ for some constant $C$. So, for large $\Lambda$, $\delta_1$ is absurdly small. Yet, this suffices for our purposes.}
\end{proof}

\begin{cor}
The set $\{\exists x: \cE_x=0\}$ is inaccessible for the limiting
equation.
\end{cor}

\appendix
\section{An averaging Theorem}\label{sec:dolgopyat}
In this appendix, for the reader convenience, 
we recall \cite[Theorem 7]{Do2} stating it in reduced generality but 
in a form directly applicable to our setting.

Let $M$ be a $\cC^\infty$ Riemannian manifold, $z\in\bR^d$ and $f_z\in {\operatorname{Diff}}^\infty(M,M)$ a family of partially hyperbolic diffeomorphisms.\footnote{By this we mean that, for each fixed $z$, at each point $x\in M$ the tangent space of $T_xM$ can be written as $E_u(x)\oplus E_c(x)\oplus E_s(x)$, where the splitting is invariant with respect to the dynamics, i.e. $d_xf E_*(x)=E_*(f(x))$ for $*\in\{u,c,s\}$. In addition, there exists constants $\lambda_1\leq \lambda_2<\lambda_3\leq\lambda_4<\lambda_5\leq\lambda_6$, with $\lambda_2,\lambda_5^{-1}<1$, such that
$\lambda_1\leq\alpha( df|_{E_s})\leq \| df|_{E_s}\|\leq \lambda_2$, $\lambda_3\leq\alpha( df|_{E_c})\leq \| df|_{E_c}\|\leq \lambda_4$ and $\lambda_5\leq\alpha( df|_{E_u})\leq \| df|_{E_u}\|\leq \lambda_6$, where $\alpha(A)=\|A^{-1}\|^{-1}$.
}
 
We say that $\{f_z\}$ is a {\em family of Anosov elements} (FAE) if there exists Abelian actions $g_{z,t}$, $t\in\bR^{d_c}$ where $d_c=\dim{E_c}$,  such that $f_z\circ g_{z,t}= g_{z,t}\circ f_z$ and ${\operatorname{span}}\{\partial_{t_i}g_{z,t}\}=E_c$.

Next, we need to discuss decay of correlations that in \cite{Do2} is meant in a very precise technical sense.
The basic concept is the one of {\em standard pairs}. For the present purposes a standard pair can be taken to be a probability measure determined by the couple $\ell=(D, \rho)$ where $D$ is a $\cC^2$ $\dim(E_u)$-dimensional manifold $D$ close to the strong unstable manifold and a smooth function $\rho\in \cC^1(D,\bR_+)$ such that $\int_D\rho=1$.\footnote{The integral is with respect the volume form on $D$ induce by the Riemannian metric.} We set $\bE_\ell(A)=\int_D A\rho$. The point is that it is possible to choose a set $\Sigma$ of manifolds $D$ of uniform bounded diameter and curvature such that, for each $D\in \Sigma$, $f_zD$ can be covered by a fixed number of elements of $\Sigma$. For each $C>0$ we consider the set $E_1=\{ (D,\rho)\;:\; D\in\Sigma, \|\rho\|_{\cC^1}(D,\bR)\leq C\}$ and let $E_2$ be the convex hull of  $E_1$ in the space of probability measures. 

It is easy to check that one can chose $\Sigma$ and $C$ such that for all $\ell\in E_1$ there exists a family $\{\ell_i\}\subset E_1$ such that $\bE_\ell(A\circ f_z)=\sum_{i=1}^{n_z} c_i^z\bE_{\ell_i}(A)$. In addition one can insure that any measure with $\cC^1$ density with respect to the Riemannian volume belongs to the weak closure of $E_2$ (see \cite{Do2} for more details).

We say that the family $\{f_z\}$ has uniform exponential decay of correlations if there exists $C_1,C_2>0$ such that, for each $z\in\bR^d$ there exists probability measures $\mu_z$ such that for each $n\in\bN$, standard pair $\ell\in E_1$ and functions $A\in\cC^1(M,\bR)$ it holds
\[
\left|\bE_\ell(A\circ f_z^n)-\mu_z(A)\right|\leq C_1e^{-C_2n}|A|_{\cC^1}.
\]

Consider now the function $F\in\cC^\infty(M\times\bR^d\times \bR_+,M\times\bR^d)$, 
\begin{equation}\label{eq:fulleq}
F(x,z,\ve)=(\tilde f(x,z,\ve),z+\ve A(x,z,\ve)),
\end{equation}
and the associated dynamical systems $F_\ve(x,z)=F(x,z,\ve)$,
such that $\tilde f(x,z,0)=f_z(x)$. Let $(x^\ve_n(x,z),z^\ve_{n}(x,z)):=F_\ve^n(x,z)$.  Then for each $g\in\cC^r(M,\bR_+)$, $\mu(g)=1$ we can define the measure $\mu_g(h):=\mu(g\cdot h)$ and consider the Dynamical Systems $(F_\ve, M\times \bR^d)$ with initial conditions $z=z_0$ and $x$ distributed according to the measure $\mu_g$. We can then view $z^\ve_{n}$ as a random variable, clearly $\bE(\psi(z^\ve_{n}))=\mu_g(\tilde\psi\circ F_\ve^n)$, where $\tilde\psi(x,z)=\psi(z)$.

\begin{thm}[\cite{Do2}] \label{thm:dolgopyat}
Let $F, F_\ve, f_z$ be defined as in \eqref{eq:fulleq} and subsequent lines. Let $f_z$ be FAE with uniform exponential decay of correlation. Suppose that there exists $\ve_0,C_r\in\bR_+$ such that $\sup_{\ve\leq\ve_0}\|A(\cdot,\cdot,\ve)\|_{\cC^r}\leq C_r$ and $\mu_z(A(\cdot,z,0))=0$ for all $z$. Also assume that $z^\ve_0=z_*$ and $x^\ve_0$ has a smooth distribution on $M$ as described above, then
\begin{enumerate}[a)]
\item The family $\{z^\ve_{\lceil t\ve^{-2}\rceil}\}$ is tight.
\item There exists functions $\sigma^2\in\cC^1(\bR^d, SL(d,\bR^d))$, $\sigma^2>0$, $a\in\cC^0(\bR^d,\bR^d)$ such that the accumulation points of $\{z^\ve_{\lceil t\ve^{-2}\rceil}\}$ are a solution of the Martingale problem associated to the diffusion
\[
\begin{split}
&dz=adt+\sigma dB\\
&z(0)=z_*,
\end{split}
\]
where $\{B_i\}_{i=1}^d$ are independent standard Brownian motions and
\[
\sigma^2(z)=\sum_{n=-\infty}^{\infty}\int_M A(x,z,0)\otimes A(f_z^nx,z,0) \mu_z(dx).
\]
Moreover $\|a\|_{\cC^0}+\|\sigma^2\|_{\cC^1}<\infty$.
\end{enumerate}
\end{thm}

\section{The properties of $\rho_{xy}$.}\label{app:rho}
Here we collect, a bit boring, proofs of the Lemmata concerning  $\rho_{xy}$.

\begin{proof}[{\bf Proof of Lemma \ref{lem:rho-tilde}}]
The non-negativity follows from the fact that the quantity is an autocorrelation, see footnote \ref{foot:auto} for details. By definition
\[
\begin{split}
\partial^n_a\partial^m_b\tilde\rho(a,b)&=\int_{-\infty}^{\infty} dt\; t^{n+m}\bE\left((L_1^{n+1}L_2^m V)\circ g^{a t}\otimes g^{bt}\cdot L_2 V\right)\\
&=(-1)^{n+m}\int_{-\infty}^{\infty} dt\; t^{n+m}\bE\left((L_1^{n+1}L_2^m V)\cdot L_2 V\circ g^{a t}\otimes g^{bt}\right).
\end{split}
\]
Applying \eqref{eq:flow-dec} to the above formula yields
\[
|\partial^n_a\partial^m_b\tilde\rho(a,b)|\leq C_{n,m}\int_{0}^{\infty} dt\; t^{n+m}e^{-c\min\{a,b\}t}\leq C_{n,m}\min\{a,b\}^{-n-m-1}.
\]
This proves the smoothness of $\tilde\rho$. To continue, consider
\[
\tilde\rho(\lambda a,\lambda b)=\int_{-\infty}^{\infty} dt\; \bE\left((L_1V)\circ g^{a\lambda t}\otimes g^{b\lambda t}\cdot L_2 V\right)=
\lambda^{-1}\tilde\rho(a,b)
\]
by the change of variables $t\to\lambda t$. The symmetry follows by a change of variables as well.
Finally,
\[
\begin{split}
\tilde\rho(a,b)&=\int_{-\infty}^{\infty} dt\; \bE\left((L_1V)\circ g^{a t}\otimes g^{bt}\cdot L_2 V\right)\\
&=\int_{-\infty}^{\infty} dt\; \bE\left(L_1V\cdot (L_2 V)\circ g^{a t}\otimes g^{bt}\right)\\
&=b^{-1}\int_{-\infty}^{\infty} dt\; \frac {d}{dt}\bE\left(L_1V\cdot  V\circ g^{a t}\otimes g^{bt}\right)\\
&\quad-\frac ab\int_{-\infty}^{\infty} dt\; \bE\left(L_1V\cdot (L_1V)\circ g^{a t}\otimes g^{bt}\right).
\end{split}
\]
The lemma follows then by the mixing of $g^{a t}\otimes g^{bt}$ (being the product of two mixing flows) and the definition of $\rho$. 
\end{proof}

To continue it is useful to define and study the function $\Gamma(\tau):=\rho(\tau,1)$.

\begin{lem}\label{lem:gamma-aus} There exists $A, B>0$ and $D\in\bR$ such that
\[
\begin{split}
&\left|\Gamma(\tau)-\frac A{1+\tau^3}\right|\leq \frac{B\tau}{1+\tau^5},\quad \forall \tau>0,\\
&|\Gamma'(\tau)-D\tau|\leq B\tau^2,\quad\quad\quad\quad \forall \tau\in(0, 1],\\
&|\Gamma'(\tau)+3A\tau^{-4}|\leq B\tau^{-5},\quad\quad \forall \tau\geq 1.
\end{split}
\]
\end{lem}
\begin{proof}
Let us start by assuming $\tau\leq 1$. By setting $\overline V(q_1)=\bE(V\;|\;q_1,v_1)$, and taking care of adding and subtracting that is needed to write convergent integrals,
\[
\begin{split}
\Gamma(\tau)=&2\int_0^\infty dt \;\bE\left(L_1V \cdot L_1V\circ g^{\tau t}\otimes g^{t})\right)\\
=&2\int_0^\infty dt \;\left[\bE\left(L_1V\cdot L_1V\circ\id\otimes g^{t}\right)-\bE((L_1 \overline V)^2)\right]\\
&+2 \int_0^\infty dt\int_{0}^{\tau t} ds\;\left[\bE\left(L_1V\cdot L_1^2V\circ g^s\otimes g^{t}\right)-\bE(L_1\overline V\cdot L_1^2\overline V\circ g^s)\right]\\
&+2 \int_0^\infty dt\;\bE(L_1 \overline V\cdot L_1 \overline V\circ g^{\tau t})\\
\end{split}
\]
The third term here vanishes since it is the variance of a coboundary. That is,
$$ \int_0^\infty dt\;\bE(L_1 \overline V\cdot L_1 \overline V\circ g^{\tau t})=
\tau^{-1} \int_0^\infty dt\;\frac d{dt}\bE(L_1 \overline V\cdot \overline V\circ g^{t})=0.$$
Also, setting $\tilde V=V-\overline V$, 
\[
\begin{split}
 &\int_0^\infty dt\int_{0}^{\tau t} ds\;\left[\bE\left(L_1V\cdot L_1^2V\circ g^s\otimes g^{t}\right)-\bE(L_1\overline V\cdot L_1^2\overline V\circ g^s)\right]\\
&=\int_0^\infty ds\int_{\tau^{-1}s}^\infty dt\;\bE\left(L_1\tilde V\cdot L_1^2\tilde V\circ g^s\otimes g^{t}\right)
=\cO\left( \int _0^\infty ds \int_{\tau^{-1} s}^\infty e^{-c t} dt\right)= \cO( \tau) 
\end{split}
\]
were we have used \eqref{eq:flow-dec} after conditioning with respect to $q_1,v_1$.
Thus\footnote{Here we use the fact that $\bE(v_1\otimes v_1\;|\;q_1, \eta)=\Id$.}
\begin{equation}\label{eq:comp2}
\begin{split}
\Gamma(\tau)=&2\int_0^\infty dt \;\left[\bE\left(\partial_{q_1}V\cdot \partial_{q_1}V\circ \id\otimes g^{t}\right)-\bE((\partial_{q_1} \overline V)^2)\right]\\
&+ \cO(\tau)=A+\cO(\tau).
\end{split}
\end{equation}

The fact that $A>0$ follows from general theory of mixing flows combined with cocycle rigidity of
geodesic flows \cite{GK, CS}.\footnote{\label{foot:auto} Indeed, for each $T>0$ and $f\in \cC^\infty$, $\bE(f)=0$,
\[
0\leq \bE\left(\left|\int_0^Tf\circ g^t \, dt\right|^2\right)=2\int_0^Tdt (T-t)\bE(f\circ g^t\cdot f)=2T\int_0^Tdt \;\bE(f\circ g^t\cdot f) +\cO(1).
\]
Thus the autocorrelation must be non negative. If it is zero then $\int_0^Tf\circ g^t \, dt$ has uniformly bounded $L^2$ norm. This implies that there exists a weakly converging subsequence to some $L^2$ function $h$ such that $\bE(h)=0$. It is easy to check that such a function is smooth in the stable direction (just compare with the average on stable manifolds) and, for each smooth function $\vf$, $\bE(hL\vf)=-\bE(f\vf)$.  Thus $\bE(h L^n \vf)=-\bE(f L^{n-1}\vf)=(-1)^n\bE(L^{n-1}f\, \vf)$, which implies $L^nh\in L^2$, i.e. $h$ is smooth along weak-stable leaves. Next, letting $\Theta(q,\nu)=(q,-\nu)$, we have $\bE(f\vf)=\bE(h\circ \Theta\cdot L\vf )$, that is $\bE((h+h\circ\Theta) L\vf)=0$ for each smooth $\vf$. In turns, this implies $h=-h\circ \Theta$ a.s.. Indeed, given $\rho\in L^2$, if $\bE(\rho)=0$ and $\bE(\rho L\vf)=0$ for all smooth $\vf$, then one can choose smooth $\rho_n$ that converges to $\rho$ in $L^2$, thus $L\rho_n$ converges weakly to zero, but then there exist convex combinations $\tilde \rho_n$ of the $\{\rho_m\}_{m\leq n}$ such that $L\tilde \rho_n$ converges to zero strongly (since the weak closure of a convex set agrees with its strong closure) and, since $L$ is a closed operator on $L^2$, it follows that $\rho$ is in the domain of $L$ and $L\rho=0$. In addition, the ergodicity of the flow implies that the only $L^2$, zero average, solution of $L\rho=0$ is $\rho=0$. Finally, since $h$ is smooth along the weak-stable foliation and $h\circ\Theta$ is smooth along the unstable foliation, then $h$ has a continuos version by the absolute continuity of the foliations and is smooth by \cite{Jou}, hence $Lh=f$. That is, if the autocorrelation is zero, then $f$ is a smooth coboundary.
At last, the claim follows since a smooth function of the coordinates only which is a coboundary must be identically zero, \cite[Corollary 1.4]{CS}. Accordingly, $\int_{-\infty}^\infty dt\; \bE\left(\partial_{q_1}V(q_1, q_2) \partial_{q_1}V(q_1, g^{t}(q_2,v_2))\;|\; q_1,v_1\right)$ must be strictly positive for positive measure set of $q_1$ otherwise, by the symmetry of the potential, the potential would be constant.}

Next, consider the case $\tau\geq 1$. By Lemma \ref{lem:rho-tilde} we have
\begin{equation}\label{eq:Gamma}
\Gamma(\tau)=\rho(\tau,1)=\tau^{-1}\rho(1,\tau^{-1})=-\tau^{-2}\tilde\rho(\tau^{-1},1)=\tau^{-3}\Gamma(\tau^{-1}).
\end{equation}
Thus $\Gamma(\tau)=\frac{A}{1+\tau^{3}}+\cO(\tau^{-4})$. This readily implies the first part of the Lemma.

Let us compute the derivative
\[
\begin{split}
\frac {\Gamma'(\tau)}2=&\int_0^\infty dt\; t\,\bE\left(L_1V\cdot L_1^2V\circ g^{\tau t}\otimes g^{t}\right)\\
=&\int_0^\infty dt\;t\int_0^{\tau t}ds\;\left\{\bE\left(L_1V \cdot L_1^3V\circ g^s \otimes g^{t}\right)-
\bE\left(L_1  V\cdot L_1^3\overline V\circ g^s\right)\right\}\\
&+\int_0^\infty dt\;t\int_0^{\tau t}ds\;\bE\left(L_1 \overline V\cdot L_1^3\overline V\circ g^s\right)\\
=&\int_0^\infty \!\!\!\!\!\!ds\int_{\tau^{-1} s}^\infty \!\!\!\!\!\!\!\!\!dt\;t\,\bE\left(L_1\tilde V \cdot L_1^3\tilde V\circ g^s \otimes g^{t}\right)+\int_0^\infty \!\!\!\!\!\!dt\;\frac{t}{\tau^2}\,\bE\left(L_1 \overline V\cdot L_1^2\overline V\circ g^t\right)\\
=& \int_0^\infty \!\!\!\!\!\!ds\int_{\tau^{-1} s}^\infty \!\!\!\!\!\!\!\!\!dt\;t\,\bE\left(L_1\tilde V \cdot L_1^3\tilde V\circ  id\otimes g^{t}\right)
-\int_0^\infty dt\;\frac{\bE\left(L_1 \overline V\cdot L_1\overline V\circ g^t\right)}{\tau^2}\\
&+ \int_0^\infty ds\int_{\tau^{-1} s}^\infty dt\;t\,\int_0^s dr\bE\left(L_1\tilde V \cdot L_1^4\tilde V\circ  g^r\otimes g^{t}\right)\\
=&- \int_0^\infty \!\!\!\!\!\!ds\int_{\tau^{-1} s}^\infty \!\!\!\!\!\!\!\!\!dt\;t\,\bE\left(L_1^2\tilde V \cdot L_1^2\tilde V\circ  id\otimes g^{t}\right)
-\frac{\bE\left(L_1 \overline V)\bE(\overline V\right)-\bE(\frac{L_1 (\overline V^2)}2)}{\tau^2}\\
&+\cO\left(\int_0^\infty ds\;\tau^{-1}s^2  e^{-cs\tau^{-1}} \right)\\
=&-\tau\int_{0}^\infty \!\!\!\!dt\;t^2\,\bE\left(L_1^2\tilde V \cdot L_1^2\tilde V\circ  id\otimes g^{t}\right)
+\cO(\tau^2)=:D\tau+\cO(\tau^2).
\end{split}
\]
On the other hand, differentiating  \eqref{eq:Gamma} yields, for $\tau$ large,
\[
\Gamma'(\tau)=-3\tau^{-4}\Gamma(\tau^{-1})-\tau^{-5}\Gamma'(\tau^{-1})=-3A\tau^{-4}+\cO(\tau^{-5})
\]
which completes the proof of the Lemma.
\end{proof}

\begin{rem} Note that $\Gamma(0)$ is not defined as the corresponding integral diverges. Nevertheless, we can set $\Gamma(0)=A$ by continuity.
\end{rem}
\begin{proof}[{\bf Proof of Lemma \ref{lem:rho-asymp}}]
Note that, by Lemma \ref{lem:rho-tilde}, 
\[
\rho(a,b)=b^{-1}\rho(ab^{-1},1)=b^{-1}\Gamma(ab^{-1}).
\] 
Hence the Lemma follows from Lemma \ref{lem:gamma-aus}.
\end{proof}

%%BIBLIOGRAPHY

\end{document}